\documentclass[12pt,english]{article}

\usepackage[english]{babel}
\usepackage{amsmath}
\usepackage{amssymb}
\usepackage{graphicx}
\usepackage{xcolor}
\usepackage{amsthm}
\usepackage{xspace}
\usepackage{authblk}
\usepackage{bm}
\usepackage[linesnumbered,ruled,vlined]{algorithm2e}
\SetKwInput{KwInput}{Input}
\SetKwInput{KwOutput}{Output}
\SetKwIF{If}{ElseIf}{Else}{if}{then}{else if}{else}{endif}

\newcommand{\trace}{\operatorname{trace}}
\newtheorem*{remark}{Remark}

\newtheorem{expl}{Example} 

\newcommand{\red}[1]{\textcolor{red}{#1}}

\newtheorem{theorem}{Theorem}[section]
\pdfoutput=1

\begin{document}
	\title{A sub-sampling algorithm preventing outliers}
	\author[1]{Deldossi, L.}
	\author[2]{Pesce, E.}
		\author[3]{Tommasi, C.}
	\affil[1]{Department of Statistical Science, Università Cattolica del Sacro Cuore, Milan, Italy}
	\affil[2]{Swiss Re Institute, Swiss Re Management Ltd, Zurich, Switzerland}
	\affil [3]{Department of Economics, Management and Quantitative Methods, University of Milan, Italy}
	\maketitle
\begin{abstract}
	Nowadays, in many different fields, massive data are available and for several reasons, it might be convenient to analyze just a subset of the data. The application of the D-optimality criterion can be helpful to optimally select a subsample of observations. However, it is well known that D-optimal support points lie on the boundary of the design space and if they go hand in hand with extreme response values, they can have a severe influence on the estimated linear model (leverage points with high influence). To overcome this problem, firstly, we propose an unsupervised  “exchange” procedure that enables us to select a “nearly” D-optimal subset of observations without high leverage values.
	Then, we provide  a  supervised version of this exchange procedure, where besides high leverage points also the outliers in the responses (that are not associated to high leverage points) are avoided. This is possible because,
	unlike other design situations, in subsampling from big datasets the response values  may be available. 
	
	Finally, both the unsupervised and the supervised selection procedures are generalized to  I-optimality, with the goal of getting accurate predictions.
\end{abstract}

\section{Introduction}
Recently, the theory of optimal design has been exploited to draw a subsample from huge datasets, containing the most information for the inferential goal; see for instance, \cite{drovandi2017, wang2019, wang2018, deldossi2020} among others. 
Unfortunately, Big Data sets usually are the result of passive observations, so some high leverage values in the covariates and/or outliers in the response variable (denoted by Y) may be present. 
The most commonly applied criterion is the D-optimality. It is well known that D-optimal designs tend to lie on the boundary of the design region thus in presence of high leverage values, all of them would be selected. Since this circumstance could have a severe influence on the estimated linear model (leverage points with high influence), in this study we propose an “exchange” procedure to select a “nearly” D-optimal subset which does not include the high leverage values. 
Avoiding high leverage points, however, does not guard from all the outliers in Y. Therefore, we also modify the previous method  to exploit the information about the responses and avoid the selection of the abnormal Y-values.
The first proposal is an unsupervised  procedure, as it is not based on the response observations, while the latter is a supervised exchange method. 
Finally, both these exchange algorithms are extended to the I-criterion, which aims at providing accurate predictions in a set of covariate-values. 

After introducing methodology and notation in Section~\ref{sec:notation}, in Section~\ref{sec:modified_exchange} we introduce the novel modified exchange algorithm to obtain both a noninformative and an informative D-optimal sample without outliers. Moreover an approach for the initialization of the above algorithms is proposed.  In Section~\ref{sec:I_optimal_sample} we adapt our proposal to the I-optimal criterion with the goal of selecting a subsample to get accurate predictions. Finally in Section~\ref{sec:simulations}  we perform some simulations which serve as motivation for the problem presented in this paper.
	
\section{Notation and motivation of the work}
	\label{sec:notation}
	
	Assume that $N$ independent responses have been generated by a super-population model
	$$
	Y_i = \bm{x}^\top_i \boldsymbol{\beta} + \varepsilon_i, \quad i=1,\ldots, N,
	$$ 
	where  $^\top$ denotes transposition, $\boldsymbol{\beta}=(\beta_0,\beta_1,\ldots,\beta_k)^\top$ is  a vector of unknown coefficients, $\bm{x}_i^\top=(1,\tilde{\bm{x}}_i^\top)$ where 
	$\tilde{\bm{x}}_i=({x}_{i1},\ldots,x_{ik})^\top$, for $i=1,\ldots, N$, are $N$  iid repetitions of a $k$-variate explanatory variable, and $\varepsilon_i$ are iid random errors with zero mean and equal variance $\sigma^2$. \\
	$\bm{D}=\{(\tilde{\bm{x}}_1^\top, Y_1),\ldots,(\tilde{\bm{x}}_N^\top, Y_N)\}$ indicates the available dataset, which is assumed to be a tall dataset, i.e. with $k << N$.
	\\
	The population under study is denoted by $U=\{1,\ldots,N\}$ and $s_n \subseteq U$ denotes a sample without replications of size $n$ from $U$ (i.e. a collection of $n$ different indices from $U$).
	Herein we describe a new sampling method from a given dataset $\bm{D}$ with the goal of selecting $n$ observations ($k < n << N$) which produce an efficient estimate of the model coefficients even in the presence of outliers.
	
	Given a sample $s_n=\{i_1,\ldots,i_n \}$, let $\bm{X}$ to be the $n \times (k+1)$ matrix whose rows are $\bm{x}_i^\top$, for $i \in s_n$, and let $\bm{Y} = (Y_{i_1}, \ldots, Y_{i_n})^\top$ be the $n \times 1$ vector of the sampled responses.
	We consider the  OLS estimator of the coefficients of the linear model based on the sample $s_n$:
	\begin{eqnarray*}
		\hat{\boldsymbol{\beta}}
		& = &
		\hat{\boldsymbol{\beta}}(s_n) = (\bm{X}^\top  \bm{X})^{-1} \bm{X}^\top  \bm{Y} \\
		& = & 
		\left(\sum_{i=1}^{N} \bm{x}_i \bm{x}_i^T I_l\right)^{-1} \sum_{i=1}^{N} \bm{x}_i \, Y_i\, I_i, 
	\end{eqnarray*}
	where 
	\begin{equation*}
	I_i = \begin{cases}
	1 & \mbox{if } i \in s_n\\
	0 & \mbox{otherwise} 
	\end{cases} \, , \qquad \text{with } i = 1, \ldots, N
	\end{equation*}
	denotes the sample inclusion indicator.
	
	To improve the precision of  $\hat{\boldsymbol{\beta}}$, we suggest to select the sample $s_n$ according to $D$-optimality. We denote the $D$-optimum sample as 
	\begin{equation*}
	s_n^* = \underset{s_n = \left\{ I_1, \ldots,I_N \right\} }{\operatorname{arg\,sup}}  \left| \sum_{i = 1}^{N} \bm{x}_i \bm{x}_i^\top I_i \right| \,.
	\end{equation*}
		
	Since the D-optimal support points usually lie in the boundary of the experimental region, when  the dataset $\bm{D}$ contains  high leverage points, $s_n^*$ includes them and if they are associated to abnormal responses then they may produce a non-reliable estimate.  Example~\ref{example_doptimal_outliers} shows how the outliers are selected by the $\bm{D}$-optimal sample.
	\begin{expl}
		\label{example_doptimal_outliers}
		An artificial dataset $\bm{D}$ with $N=10000$ observations has been generated from a simple linear model,   
		\begin{equation*}
		Y_{i} = \beta_0 + \beta_1 x_i + \varepsilon_i, \quad i = 1, \ldots, N,
		\end{equation*}
		in the following way:\\
		for $i = 1, \ldots, 9990$, $\boldsymbol{\beta} = (1.5, 2.7)^\top$,
		$x_i \sim \mathcal{N}(3, 4)$, $\varepsilon_i \sim \mathcal{N}(0, 9^2)$;\\
		for $i = 9991, \ldots,10000$, $\boldsymbol{\beta} = (1.5, -2.7)^\top$
		$x_i \sim \mathcal{N}(3, 20)$, $\varepsilon_i \sim \mathcal{N}(0, 20^2)$.\\ 
		 The left-hand side of Figure~\ref{fig:example_doptimal_outliers} displays these last 10 observations in red, while the majority of the data, generated from the first distribution,  are displayed in black.
		The right-hand side of Figure~\ref{fig:example_doptimal_outliers} emphasises the D-optimal subsample  of size $n =100$, $s_n^*$, displaying its support points in blue.  
		As expected, all the abnormal values in $X$ are included in $s_n^*$ because they maximize the determinant of the information matrix ($s_n^*$ has been obtained by applying the  function \texttt{od\_KL} of the \texttt{R} package \texttt{OptimalDesign}~\cite{harman2019optimaldesign}). 
		
		\begin{figure}
			\centering
			\includegraphics[width=0.9\textwidth]{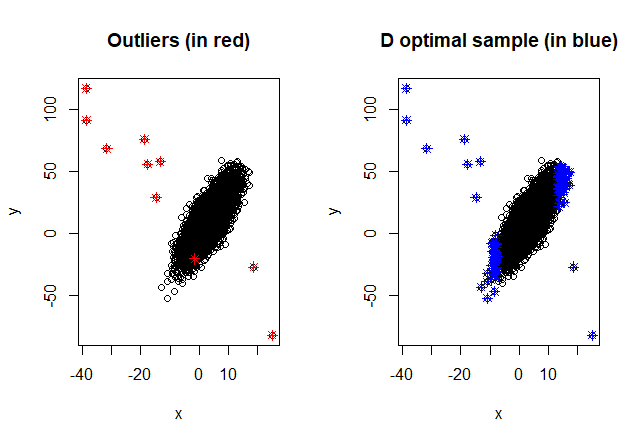}
			\caption{Outliers (in red) and the $D$-optimal sample (in blue). \label{fig:example_doptimal_outliers}}
		\end{figure}	
	\end{expl}
	To avoid the outliers when applying  D-optimal subsampling, we propose  a modification of the well known exchange algorithm.
	Before describing our proposal, we recall that an observation  $\bm{x}_i$ with $i=1, \cdots, n$ is called an \textit{high leverage point} when its leverage value   
$h_{ii} =\bm{x}^\top_i (\bm{X}^\top  \bm{X})^{-1} \bm{x}_i $ is greater than a threshold, i.e. $$h_{ii} >\nu_1\,\frac{k+1}{n} $$
where $\nu_1$ is a tuning parameter usually set equal to 2 ~\cite{hoaglin1978}.\\ 
 A high leverage point can be either \textit{good} or \textit{bad}:  if its  associated response is \lq\lq abnormal'', it is bad  because,  in this case, it might alter the  model fitted by the bulk of the data; otherwise, if the response is not an outlier, the high leverage point is good because it would reduce the variance of the parameters' estimates. 
    \section{Modified Exchange Algorithms}
    \label{sec:modified_exchange}
    The common structure of the $t$-th iteration of an exchange algorithm consists in adding a point $\bm{x}_{j_a}$ (chosen from a list of candidate points ${\cal C}^{(t)}$) to the available sample $s_n^{(t-1)}$ and then delete a point from it. The choice of the augmented and deleted points is based on the achievement of some optimality criterion. For D-optimality, the augmented observation $\bm{x}_{j_a}$ is the $x$-value of the unit with the largest prediction variance, 
    $$
    j_a=\arg\max_{j\in {\cal C}^{(t)}}\bm{x}^\top_j (\bm{X}^\top  \bm{X})^{-1} \bm{x}_j, 
    $$
    and is chosen from the experimental domain, if replications are admitted, or from  the complementary space of $s_n^{(t-1)}$, if it is not.
    The deleted point is that with the smallest prediction variance, i.e. with the smallest leverage value $h_{ii}$, $i=1,\ldots,n+1$ (see Chp. 12 in \cite{atkinson2007}).
    Our main idea is to modify the algorithm in such a way that points with high leverage scores are not proposed for the exchange, avoiding  the dangerous combination of high leverage scores and abnormal responses. This goal is reached:
        a) by  switching the augmentation and deletion steps; b) by changing the set ${\cal C}^{(t)}$ where the observation to be added is searched.
    
    In step b), if the information about the responses is not used to identify the set  ${\cal C}^{(t)}$, then 
     the modified D-optimal sample 
     is non-informative for the parameters of interest.  This unsupervised procedure is described in detail in Subsection \ref{sec:sampling_strategy}. 
     
      Avoiding high leverage points, however, does not guard from all the outliers in $Y$;  points may be present which are in the core of the data wrt the auxiliary variables but are abnormal wrt the response variable. 
      In Subsection \ref{informative sample} we propose a supervised version of the algorithm, where in step b) we exploit the response values to remove the outliers in $Y$. Let us note that the obtained optimal sample becomes informative because it depends on the responses. 
	\subsection{Noninformative D-optimal samples without high leverage points}
	\label{sec:sampling_strategy}
	Let $s_n^{(0)}$ be an initial sample of size $n$, which does not include high leverage points. At the end of this section we describe a method for getting such an initial sample.
	
	Let $ {\cal C}^{(t)}$ be a set of candidate points for the exchange at the current iteration and let  ${\bm{X}}_t$ be the design matrix associated to the sample $s_n^{(t)}$. 
	To update  $s_n^{(t)}$, firstly  we remove from it the unit $i_{m}$  with the smallest prediction variance, i.e.
	$$
	i_{m} = \underset{i \in s_n^{(t)}}{\arg \min} ~ h_{ii};
	$$
	let ${\bm{X}}_t^{\!-}$ denote the design matrix obtained by omitting the row  $\bm{x}_{i_m}$ from $\bm{X}_t$.
	Subsequently, we add the unit $j_a \in {\cal C}^{(t)} $  which presents the largest prediction variance  
	${\bm{x}}_{j_a}^\top ({{\bm{X}}_t^{\!-}}^\top  {\bm{X}}_t^{\!-})^{-1} \bm{x}_{j_a}$, where 
	\begin{equation}
	\label{candidate_p}
	{\cal C}^{(t)}= \left\{j:\;\;
	h_{i_{m} i_{m}}< h_{i_m i_m}(\bm{x}_j) < \nu_1 \,\frac{k+1}{n}, 
	\right\}
	\end{equation}
	and 
	$h_{i_m i_m}(\bm{x}_j)$ 
	is the leverage score obtained exchanging $\bm{x}_{i_m}$ with $\bm{x}_j$ for $j\in \{U-s_n^{(t)}\}$. For computational purposes, let us note that from Saerle (1982, p.153)
 \begin{equation}
 \label{inv_X-}
 ({\bm{X}_t^{\!-}}^\top  \bm{X}_t^{\!-})^{-1}=(\bm{X}_t^\top  \bm{X}_t)^{-1}
+(\bm{X}_t^\top  \bm{X}_t)^{-1}  \frac{\bm{x}_{i_m} \bm{x}_{i_m}^T}{1- \bm{x}_{i_m}^T (\bm{X}_t^\top  \bm{X}_t)^{-1} \bm{x}_{i_m}} (\bm{X}_t^\top  \bm{X}_t)^{-1}
\end{equation}
 in addition,
	the next theorem provides an analytical expression for $h_{i_m i_m}(\bm{x}_j)$.
\begin{theorem}
	Let  $\widetilde{\bm{X}}_t$ be the design matrix obtained from $\bm{X}_t$ exchanging $\bm{x}_{i_m}$ with $\bm{x}_j$,	
  then
	\begin{eqnarray}
	\label{h_xj}
	h_{i_m i_m}(\bm{x}_j)= \bm{x}^\top_j \left(\widetilde{\bm{X}}_t^\top  \widetilde{\bm{X}}_t\right)^{-1} \bm{x}_j
	\end{eqnarray}
	where 
	\begin{equation}
	\label{Inv_s}
	\left(\widetilde{\bm{X}}_t^\top  \widetilde{\bm{X}}_t\right)^{-1}= 
	(\bm{X}_t^\top  \bm{X}_t)^{-1}-(\bm{X}_t^\top  \bm{X}_t)^{-1} \,\frac{\bm{A}}{d}\, (\bm{X}_t^\top  \bm{X}_t)^{-1},
	\end{equation}
	with
	\begin{eqnarray*}
		\bm{A}&=& \bm{x}_{i_m}^\top 
		(\bm{X}_t^\top \bm{X}_t)^{-1} 
		\bm{x}_{j} \,
		(\bm{x}_{j} \bm{x}_{i_m}^\top+\bm{x}_{i_m} \bm{x}_{j}^\top)
		+ [1-\bm{x}_{i_m}^\top (\bm{X}_t^\top \bm{X}_t)^{-1} \bm{x}_{i_m}]\,\bm{x}_{j}\bm{x}_{j}^\top\red{+}\\
		&-&
		[1+\bm{x}_{j}^\top (\bm{X}_t^\top \bm{X}_t)^{-1} \bm{x}_{j}]\,\bm{x}_{i_m}\bm{x}_{i_m}^\top;\\[.2cm]
		d &=&
		[1-\bm{x}_{i_m}^\top (\bm{X}_t^\top \bm{X}_t)^{-1} \bm{x}_{i_m}]\,[1+\bm{x}_{j}^\top (\bm{X}_t^\top \bm{X}_t)^{-1} \bm{x}_{j}]+[\bm{x}_{i_m}^\top 
		(\bm{X}_t^\top \bm{X}_t)^{-1} 
		\bm{x}_{j} ]^2.
	\end{eqnarray*}
\end{theorem}
\begin{proof}
	Expression (\ref{Inv_s}) can be obtained from Lemma 3.3.1 in Fedorov (1972) after some cumbersome algebra.
\end{proof}
		
	In force of the upper bound in (\ref{candidate_p}), our proposal is to consider as candidates for the exchange only observations in $\{U-s_n^{(t)}\}$ which  are not high leverage points.
	In addition, to speed up the algorithm we reduce the number of exchanges by imposing the lower bound in (\ref{candidate_p}). Without this lower bound, whenever $h_{i_m i_m}(\bm{x}_j) \leq h_{i_{m} i_{m}}$, the new observation $j$ would be immediately removed at the subsequent step of the algorithm.
	
	Algorithm 1  describes in detail all the steps to select a D-optimal sample without high leverage points; instead Algorithm 2 illustrates how to select an initial sample $s_n^{(0)}$ to start Algorithm 1.
	\begin{algorithm}[!h]
		{
			\caption{Non-informative D-optimal sample without high leverage points}
			\label{Algorithm 1}
			\SetAlgoLined
			\vspace{0.1cm}
			\KwInput{Dataset $\bm{D}$, sample size $n$, initial sample $s_n^{(0)}$, $\nu_1$} 
			\KwOutput{D-optimal sample without high leverage points}
			Set $t=0$\;
			\While{$t<t_{max}$}
			{Compute the leverage scores for the current sample $h_{ii}={\bm{x}}_{i}^\top (\bm{X}_{t}^\top                {\bm{X}_{t}})^{-1} {\bm{x}_i},$
				where $\bm{X}_{t}$ is the $n \times k$ matrix whose rows are $\bm{x}_i^T$ with $i \in s_n^{(t)}$\;
				Identify unit 
				$i_{m} = \underset{i \in s_n^{(t)}}{\arg \min} ~ h_{ii}$\;
					Compute 
				$
				(\bm{X}_{t}^{\!-\top} \bm{X}_{t}^{\!-})^{-1}=(\bm{X}_{t}^\top  \bm{X}_{t})^{-1}
				+(\bm{X}_{t}^\top  \bm{X}_{t})^{-1}  \frac{\bm{x}_{i_m} \bm{x}_{i_m}^T}{1- \bm{x}_{i_m}^T (\bm{X}_{t}^\top  \bm{X}_{t})^{-1} \bm{x}_{i_m}} (\bm{X}_{t}^\top  \bm{X}_{t})^{-1}
				$\;
				Select randomly $\tilde N \leq N-n$ units from $\Big\{ U-s_n^{(t)} \Big\}$. Let $\bm{x}_j$, with $j = 1, \ldots, \tilde{N}$, the observations for these units\;
				From (\ref{Inv_s}), compute $(\widetilde{\bm{X}}_t^\top \widetilde{\bm{X}}_t)^{-1}$ and determine 
				the leverage scores $h_{i_m i_m}(\bm{x}_j)=\bm{x}_j^\top (\widetilde{\bm{X}}_t^\top \widetilde{\bm{X}}_t)^{-1} \bm{x}_j$\;
				Identify the set of candidate points for the exchange with $i_{m}$:
				${\cal C}^{(t)}=\left\{j:\;\;
				h_{i_{m} i_{m}}< h_{i_m i_m}(\bm{x}_j) < \nu_1 \,\frac{k+1}{n} 
				\right\}$\;
				Select from ${\cal C}^{(t)}$ the observation $j_a = \underset{j \in {\cal C}^{(t)}}{\operatorname{arg\,max}} ~ 
				\bm{x}_j^\top (\bm{X}_{t}^{\!-\top} \bm{X}_{t}^{\!-})^{-1} \bm{x}_j$\;
				Update $s_n^{(t)}$ by replacing unit $i_m$ with $j_a$, to form  $s_n^{(t+1)}$\;
				Set $t=t+1$\;
			}
		}
\vspace{0.1cm}
	\begin{remark}
		In step \textit{6} it is reasonable to consider  $\tilde{N}=N-n$ whenever $N$ is not too large.
	\end{remark}	
	\end{algorithm}	
	\begin{algorithm}[!h]
		{
			\caption{Initialization step for Algorithm 1}
			\label{Algorithm 2}
			\SetAlgoLined
			\vspace{0.1cm}
			\KwInput{Dataset $\bm{D}$, sample size $n$, $\nu_2$} 
			\KwOutput{$s_n^{(0)}$: initial sample  without high leverage points}
			From $U$ select without replacement a simple random sample of size $n$,  $r_n^{(0)}$ and set $t=0$\;
			\While{$t<t_{max}$}
			{Compute the leverage scores for the current sample $h_{ii}={\bm{x}}_{i}^\top (\bm{X}_{t}^\top {\bm{X}_{t}})^{-1} {\bm{x}_i},$
				where $\bm{X}_{t}$ is the $n \times k$ matrix whose rows are $\bm{x}_i^T$ with $i \in r_n^{(t)}$\;
				Identify unit 
				$i_{m} = \underset{i \in s_n^{(t-1)}}{\arg \max} ~ h_{ii}$ \;
				{\eIf{$h_{i_{m} i_{m}}<  \nu_2 \,\frac{k+1}{n}$ } {set $s_n^{(0)}=r_n^{(t)}$ and stop the iterative procedure} {
				Select randomly $\tilde N \leq N-n$ units from $\Big\{ U-r_n^{(t)} \Big\}$. Let $\bm{x}_j$, with $j = 1, \ldots, \tilde{N}$, the observations for these units. From (\ref{h_xj}), compute  $h_{i_m i_m}(\bm{x}_j)$ and identify the set of  points candidate for the exchange with $i_{m}$:
				${\cal C}^{(t)}=\left\{j:\;\;
				 h_{i_m i_m}(\bm{x}_j) < \nu_2 \,\frac{k+1}{n} 
				\right\}$\;
				Select at random a unit $j_a$ from  ${\cal C}^{(t)}$ \;
				Determine  $r_n^{(t+1)}$ by replacing unit $i_m$ with $j_a$ in $r_n^{(t)}$ \;
				Compute $(\bm{X}_{t+1}^\top {\bm{X}_{t+1}})^{-1}$ by applying (\ref{Inv_s}) which is based on $(\bm{X}_{t}^\top {\bm{X}_{t}})^{-1}$\;
				Set $t=t+1$\;
			}
		} 
			}  
		}
		\vspace{0.1cm}
		\begin{remark}
			In step \textit{8} it is reasonable to consider  $\tilde{N}=N-n$ whenever $N$ is not too large.
		\end{remark}	
	\end{algorithm}	
	\subsection{Informative D-optimal sample without outliers}
	\label{informative sample}
	The previous exchange algorithm can be applied whenever
	the response values are not observed (for instance, if it is expensive to measure $Y$); this selection procedure protects against potential outliers in $Y$ that are associated with the high leverage points in the factor-space. However, if  the response values are available, then this information can be exploited by the exchange algorithm to avoid all the outliers in $Y$, obtaining an informative D-optimal subsample.
	
	According  to \cite{Chatterjee1986influential} an influential data point is an observation that strongly influences the fitted values. To identify these influential data points, Cook’s distance may be applied. In fact, Cook's distance  for the $i$-th observation, $C_i$, measures how much all of the fitted values in the model change when the $i$-th data point is deleted:
	\begin{eqnarray}
		C_i&=&\dfrac{(\hat{\bm{Y}}-\hat{\bm{Y}}_{(i)})^\top (\hat{\bm{Y}}-\hat{\bm{Y}}_{(i)})}{(k+1)\hat\sigma^2},
		\nonumber
		\\
		&=&
		\frac{(Y_i-\hat{Y}_i)^2}{(k+1) \hat\sigma^2} \cdot \frac{h_{ii}}{(1-h_{ii})^2},
		\quad 
		i=1,\ldots,n,
		\label{Cook_distance}
	\end{eqnarray}
	where $\hat{\bm{Y}}= \bm{X} \hat{\boldsymbol{\beta}}^\top$, $\hat\sigma^2$ is the residual mean square estimate of $\sigma^2$ and $\hat{\bm{Y}}_{(i)}= \bm{X} \hat{\boldsymbol{\beta}}_{(i)}^\top$  is the vector of predicted values when the $i$-th unit is removed from the data set $\bm{D}$. 
	
	 A general practical rule is that any observation with a Cook’s distance larger than $4/n$ may be considered an influential point.
	 
	 When the response values are available, Algorithm 1 can be improved by removing the influential points from the  set ${\cal C}^{(t)}$ of the data candidate for the exchange, as illustrated in Algorithm 3 (for the computation of Cook's distance, expression  (\ref{Cook_distance}) is used to reduce the computational burden).
	

\begin{algorithm}[!h]{
		\caption{Informative D-optimal sample without outiliers: additional steps to be included between 9 and 10 in Algorithm 1 }
		\label{Algorithm 3}
		\SetAlgoLined
		\setcounter{AlgoLine}{0}
		\vspace{0.1cm}
		\KwInput{Dataset $\bm{D}$, sample size $n$} 
		\KwOutput{Informative D-optimal sample without outliers}
		Compute Cook's distance for unit $j_a$,  $C_{j_a}$, from (\ref{Cook_distance})\;
		{\eIf{$C_{j_a} < 4/n$ 
		} {accept the exchange and go to step {\bf 10} of Algorithm 1}  {reject the exchange; remove unit $j_a$ from ${\cal C}^{(t)}$ and go back to step {\bf 9} of Algorithm 1}
		} 
	}
\end{algorithm}

\section{Optimal subsampling to get accurate predictions}
	\label{sec:I_optimal_sample}
In the previous section we aim at selecting a subsample with the goal of getting a precise estimation of the parameters. Differently, if we are interested in obtaining accurate predictions on a set of values ${\cal X}_0=\{ \bm{x}_{01},\ldots,\bm{x}_{0N_0}\}$, then we should 	select the observations minimizing the overall prediction variance. 
Let $\hat{Y}_{0i}= \hat{\boldsymbol{\beta}}^\top \bm{x}_{0i}$ be the prediction of $\mu_{0i}=E(Y_{0i}|\bm{x}_{0i})$ at $\bm{x}_{0i}$, $i=1,\ldots,N_0$. 
 The prediction variance at $\bm{x}_{0i}$, also known as \lq\lq mean squared prediction error'' is
	$${\rm MSPE}(\hat{Y}_{0i}|\bm{x}_{0i},\bm{X})={\rm E}[(\hat{Y}_{0i}-\mu_{0i})^2|\bm{x}_{0i},\bm{X}].$$
If ${\bm{X}}_0$ is the $N_0\times k$ matrix whose $i$-th row is $\bm{x}_{0i}^\top$, then a measure of the overall mean squared prediction error is the sum of the prediction variances in ${\cal X}_0$:
\begin{eqnarray}
&\sum_{i=1}^{N_0} &  \!\!
{\rm MSPE}(\hat{Y}_{0i}|\bm{x}_{0i},\bm{X}) 
\,=\,
 \sigma^2 \, 
{\rm trace} [{\bm{X}}_0 ({\bm{X}}^\top {\bm{X}})^{-1} {\bm{X}}_0^\top]\!
 \nonumber\\
&=&\!\!\!
\sigma^2  \,\trace \!\left[\left( \sum_{i = 1}^{N} \bm{x}_i \bm{x}_i^\top I_i\right)^{-1}  \!
		 {\bm{X}}_0^\top {\bm{X}}_0  \right]\!,
		\quad
		I_i = \begin{cases}
			1 & \mbox{if } i \in s_n\\
			0 & \mbox{otherwise} 
		\end{cases}
\label{MSPE}
\end{eqnarray}
with $l = 1, \ldots, N$.
In this context, the I-optimal sample should be selected, which minimizes the  overall prediction variance (\ref{MSPE}):
$$
		s_n^{I}
		= 
		\underset{s_n = \left\{ I_1, \ldots,I_N \right\} }{\operatorname{arg\,inf}} \trace \left[\left( \sum_{i = 1}^{N} \bm{x}_i \, \bm{x}_i^\top I_i\right)^{-1}  \!
		 {\bm{X}}_0^\top {\bm{X}}_0  \right].
$$
	
	If we also aim at preventing outliers, then we have to modify the deletion and augmentation steps of the exchange algorithm described in Section \ref{sec:sampling_strategy} accordingly to the $I$-criterion. By taking into account the results given in Appendix A of Meyer and Nachtsheim (1995),  the current sample $s_n^{(t)}$ should be updated by removing the unit $i_{m}$
		which minimises
		$$\tilde{h}_{ii}=\frac{  \bm{x}_i^\top \,({\bm{X}}_t^\top {\bm{X}_t})^{-1}
			\;{\bm{X}}_0^\top {\bm{X}}_0\;
			({\bm{X}}_t^\top {\bm{X}_t})^{-1}\,
			\bm{x}_i }
		{1- \bm{x}_i^T ({\bm{X}}_t^\top {\bm{X}}_t)^{-1} \bm{x}_i },
		$$
		where $\bm{X}_t$ is the $n \times k$ matrix whose rows are $\bm{x}_i^T$ with $i \in s_n^{(t)}$.  
		
		Subsequently, from a set ${\cal C}^{(t)}$ of candidate points,  we should add the unit 
		$$
		j_a=\arg\max_{j\in {\cal C}^{(t)}}\frac{{  \bm{x}_j^\top \,({\bm{X}}_t^{\!-\top} {\bm{X}}_t^{\!-})^{-1}
				\;{\bm{X}}_0^\top {\bm{X}}_0\;
				({\bm{X}}_t^{\!-\top} {\bm{X}}_t^{\!-})^{-1}\,
				\bm{x}_j }}
		{1+ \bm{x}_j^T ({\bm{X}}_t^{\!-\top} {\bm{X}}_t^{\!-})^{-1} \bm{x}_j },
		$$
		where ${\bm{X}}_t^{\!-}$
		is the design matrix obtained by removing the row  $\bm{x}_{i_m}$ from $\bm{X}_t$ and $({\bm{X}}_t^{\!-\top} {\bm{X}}_t^{\!-})^{-1}$
		can be computed from (\ref{inv_X-}). 
		The set of candidate points should be formed by units that are not immediately removed in the subsequent step of the procedure and also are not high leverage points; therefore, ${\cal C}^{(t)}$ is 
		$$
		{\cal C}^{(t)}=\left\{j:\;\;
		\tilde{h}_{i_m i_m}(\bm{x}_j) >\tilde{h}_{i_{m} i_{m}}\, \cap \,
		{h}_{i_m i_m}(\bm{x}_j)< \nu_1 \,\frac{k+1}{n} 
		\right\}
		$$
		where
		$$
		\tilde{h}_{i_m i_m}(\bm{x}_j)=
		\frac{  \bm{x}_j^T \,(\widetilde{\bm{X}}_t^\top \widetilde{\bm{X}}_t)^{-1}
		\;{\bm{X}}_0^\top {\bm{X}}_0\;
			(\widetilde{\bm{X}}_t^\top \widetilde{\bm{X}}_t)^{-1}\,
			\bm{x}_j }
		{1- \bm{x}_j^T (\widetilde{\bm{X}}_t^\top \widetilde{\bm{X}}_t)^{-1} \bm{x}_j },
		$$ 
		$$
		{h}_{i_m i_m}(\bm{x}_j)=\bm{x}_j^T (\widetilde{\bm{X}}_t^\top \widetilde{\bm{X}}_t)^{-1} \bm{x}_j, 
		$$
		$\widetilde{\bm{X}}_t$ is the matrix obtained from $\bm{X}_t$ by exchanging $\bm{x}_{i_m}$ with $\bm{x}_j$ and $(\widetilde{\bm{X}}_t^\top \widetilde{\bm{X}}_t)^{-1}$ can be computed from (\ref{Inv_s}).
		\begin{algorithm}[]
		{
			\caption{Non-informative I-optimal sample without high leverage points}
			\label{Algorithm 4}
			\SetAlgoLined
			\vspace{0.1cm}
			\KwInput{Dataset $\bm{D}$, sample size $n$, initial sample $s_n^{(0)}$, prediction-set ${\cal X}_0=\{ \bm{x}_{01},\ldots,\bm{x}_{0N_0}\}$, $\nu_1$} 
			\KwOutput{I-optimal sample without high leverage points}
			Set $t=0$\;
			\While{$t<t_{max}$}
			{For the current sample, compute
			$\tilde{h}_{ii}=\dfrac{  \bm{x}_i^\top \,({\bm{X}}_t^\top {\bm{X}_t})^{-1}
			\;{\bm{X}}_0^\top {\bm{X}}_0\;
			({\bm{X}}_t^\top {\bm{X}_t})^{-1}\,
			\bm{x}_i }
		{1- \bm{x}_i^T ({\bm{X}}_t^\top {\bm{X}}_t)^{-1} \bm{x}_i },$
				where $\bm{X}_{t}$ is the $n \times k$ matrix whose rows are $\bm{x}_i^T$ with $i \in s_n^{(t)}$ and 
				${\bm{X}}_0$ is the $N_0\times k$ matrix whose rows are the elements of ${\cal X}_0$
				\;
				Identify unit 
				$i_{m} = \underset{i \in s_n^{(t)}}{\arg \min} ~ \tilde{h}_{ii}$\;
				Compute 
				$
				(\bm{X}_{t}^{\!-\top} \bm{X}_{t}^{\!-})^{-1}\!=\!(\bm{X}_{t}^\top  \bm{X}_{t})^{-1}
				\!+\!  \dfrac{(\bm{X}_{t}^\top  \bm{X}_{t})^{-1}\,
				\bm{x}_{i_m} \bm{x}_{i_m}^T \,
				(\bm{X}_{t}^\top  \bm{X}_{t})^{-1}}
				{1- \bm{x}_{i_m}^T (\bm{X}_{t}^\top  \bm{X}_{t})^{-1} \bm{x}_{i_m}}
				$
				\;
				Select randomly $\tilde N \leq N-n$ units from $\Big\{ U-s_n^{(t)} \Big\}$. Let $\bm{x}_j$, with $j = 1, \ldots, \tilde{N}$, the observations for these units\;
				From (\ref{Inv_s}) compute $(\widetilde{\bm{X}}_t^\top \widetilde{\bm{X}}_t)^{-1}$ and  determine the leverage scores 
				$h_{i_m i_m}(\bm{x}_j)=\bm{x}_j^\top (\widetilde{\bm{X}}_t^\top \widetilde{\bm{X}}_t)^{-1} \bm{x}_j $ and
				$
		\tilde{h}_{i_m i_m}(\bm{x}_j)=
		\dfrac{  \bm{x}_j^T \,(\widetilde{\bm{X}}_t^\top \widetilde{\bm{X}}_t)^{-1}
			\,{\bm{X}}_0^\top {\bm{X}}_0\,
			(\widetilde{\bm{X}}_t^\top \widetilde{\bm{X}}_t)^{-1}\,
			\bm{x}_j }
		{1- \bm{x}_j^T (\widetilde{\bm{X}}_t^\top \widetilde{\bm{X}}_t)^{-1} \bm{x}_j }
		$
		\;
		Identify the set of candidate points for the exchange with $i_{m}$:
		$
		{\cal C}^{(t)}=\left\{j:\;\;
		\tilde{h}_{i_m i_m}(\bm{x}_j) >\tilde{h}_{i_{m} i_{m}}\, \cap \,
		{h}_{i_m i_m}(\bm{x}_j)< \nu_1 \,\frac{k+1}{n} 
		\right\}
		$\;
		Select from ${\cal C}^{(t)}$ the observation $
		j_a=\arg\max_{j\in {\cal C}^{(t)}}\dfrac{{  \bm{x}_j^\top \,({\bm{X}}_t^{\!-\top} {\bm{X}}_t^{\!-})^{-1}
				\;{\bm{X}}_0^\top {\bm{X}}_0\;
				({\bm{X}}_t^{\!-\top} {\bm{X}}_t^{\!-})^{-1}\,
				\bm{x}_j }}
		{1+ \bm{x}_j^T ({\bm{X}}_t^{\!-\top} {\bm{X}}_t^{\!-})^{-1} \bm{x}_j }
		$\;
				Update $s_n^{(t)}$ by replacing unit $i_m$ with $j_a$, to form  $s_n^{(t+1)}$\;
				Set $t=t+1$\;
			}
		}
\vspace{0.1cm}
	\end{algorithm}	
\section{Numerical studies}
	\label{sec:simulations}
	
\subsection{Simulation results}
	
In this section, we evaluate the performance of our proposals through a simulation study.
We start from the random generation of $H=30$ datasets of size $N=10^6$, each one including $N_2=500$ high leverage points/outliers. The computation of some metrics will illustrate the validity of our procedure in selecting a D- or I-optimal subsample without outliers.

Precisely, for $h=1, \ldots, H$,  $N$  iid repetitions of a $10$-variate explanatory variable 
	$_h\tilde{\bm{x}}_i=({x}_{i1},\ldots,x_{i10})^\top$ are generated as follows: 
\begin{enumerate}
   \item
 for $j=1,\ldots,3$, $x_{ij}$ are independently distributed as $U(0,5)$;
 \item
 for $j=4,\ldots,7$, $x_{ij}$  are distributed as a multivariate normal r.v. with zero mean and:\\
 2.a) for $  i=1, \ldots, (N-N_2)$, covariance matrix $
\bm{\Sigma}_1=\left[\begin{array}{cc}
9 & -1 \\
-1 & 9\\
\end{array}  \right]$ \\
2.b)  for $i=(N-N_2)+1, \ldots, N$, covariance matrix $
\bm{\Sigma}_{1.out}=\left[\begin{array}{cc}
25& 1 \\
1 & 25\\
\end{array}  \right]$; 
\item
 for $j=8,9$,  $x_{ij}$ are distributed as a multivariate  t-distribution with $3$ degrees of freedom and scale matrix
$
\bm{\Sigma}_2=\left[\begin{array}{cc}
1 & 0.5 \\
0.5 & 1\\
\end{array}  \right]
$; 
\item
for $j=10$, $x_{ij}$ is distributed as a Poisson distribution ${\cal P}(5)$.
\end{enumerate}
For each $N\times (k+1)$ factor-matrix $_h\bm{X}$, whose $i$-th raw is $_h{\bm{x}}^\top_i=(1,  _h{\tilde{\bm{x}}}^\top_i)$ ($i=1,\ldots,N$), we have generated  $S=50$ independent $N\times 1$ response vectors
$_h{\!\bm{Y}\!}_{s}$ (with $s=1,\ldots, S$), whose $i$-th item is  
$$
_h{Y}_{s,i} = _h\!\bm{x}^\top_i \boldsymbol{\beta} + \varepsilon_{si}, \quad i=1,\ldots, N, 
$$  
with 
\begin{itemize}
    \item[i)]
     $\bm{\beta}=(1, 1, 1, 1, 2, 2, 2, 2, 1, 1, 1)$ and $\sigma=3$ for $i=1, \ldots, N-N_2$ 
    \item[ii)]
	 $\bm{\beta}\!=\!(1, 1, 1, 1, -2, -2, -2, -2, 1, -1, -1)$, $\sigma\!=\!20$  for $i\!=\!N\!-\!N_2\!\ +1, \ldots, N$. 
\end{itemize}
At each simulation step $(h,s)$, with $h=1,\ldots,H$ and $s=1,\ldots,S$, 
we have applied the following Algorithms:
\begin{enumerate}
    \item Non-informative I (Algorithm 4)
    \item Non-informative D (Algorithm 1) 
    \item  Informative I (Algorithm 4 and Algorithm 3)
    \item Informative D (Algorithm 1 and Algorithm 3)
    \item Simple random sampling (SRS): passive learning selection
\end{enumerate}
to draw a different subsample from the simulated dataset: 
$$_h\bm{D}_{s}=\{(_h{\bm{x}}_1^\top,\, _{h}{y}_{s,1}),\ldots,(_h{\bm{x}}_N^\top,\, _{h}{y}_{s,N})\},\quad h=1, \ldots,H,\; s=1,\ldots,S. 
$$

To check the validity of the inferential results obtained from the distinct subsamples, we have generated a test set of size $N_T=500$:  $$\bm{D}_{T}=\{(\bm{x}_{T1},y_{T1}),\ldots, (\bm{x}_{TN_{T}},y_{TN_{T}}) \},$$ without high leverage points and outliers (i.e. with $N_2=0$). 

Finally, to implement the I-optimality procedure, we have generated a prediction region ${\cal X}_0$ 
	without high leverage points; in addition, to compare the behaviour of the distinct subsamples in ${\cal X}_0$, we have generated also the corresponding responses (without outliers). Let 
	$$\bm{D}_{0}=\{(\bm{x}_{01},y_{01}),\ldots, (\bm{x}_{0N_{0}},y_{0N_{0}}) \}$$
	be the prediction set, where $N_0=500$. 
\\[.2cm]
Let us denote by $s_{n}^{(h,s)}$ a subsample selected from the dataset $_h\bm{D}_{s}$ generated at the $(h,s)$-th simulation step, for $h=1,\ldots,H$ and $s=1,\ldots,S$, and 
let
  $$
  I_{i}^{(h,s)}\! = \begin{cases}
			1 & \mbox{if } i \in s_n^{(h,s)}\\
			0 & \mbox{otherwise} 
		\end{cases}, \quad i=1,\ldots,N,
$$
be the corresponding sampling indicator variable.\\
At each simulation step $(h,s)$,  
to evaluate the performance of the subsampling techniques, we have computed:
\begin{itemize}
    \item 
   The average mean squared prediction error in ${\cal X}_0$ (from (\ref{MSPE})): 
$$
  {\rm MSPE}_{{\cal X}_0}^{(h,s)}=
\sigma^2  \dfrac{\trace \!\left[\!\left( \sum_{i = 1}^{N} {_h\bm{x}}_i\, {_h\bm{x}}_i^\top I_{i}^{(h,s)}\right)^{\!-1}  \!
		 \bm{X}_0^\top \bm{X}_0  \right]}{N_0};
$$  
   \item 
    The logarithm of the determinant of the information matrix: 
$$
	{\rm Log(det)}^{(h,s)} =  \log \left| \sum_{i = 1}^{N} {_h\bm{x}}_{i}\, {_h\bm{x}}_{i}^\top I_{i}^{(h,s)} \right|; 
$$
    \item 
    The average squared prediction error in ${\cal X}_0$ and in ${\cal X}_{T}=\{\bm{x}_{T1},\ldots,\bm{x}_{TN_{T}}\}$:
    $$
{\rm SPE}_{{\cal X}_{0}}^{(h,s)}=  \dfrac{\sum_{i=1}^{N_{0}} ( 
\hat{y}_{0i}^{(h,s)}- \mu_{0i})^2}{N_{0}}
\;\; {\rm and} \;\; 
{\rm SPE}_{{\cal X}_{T}}^{(h,s)} =  \dfrac{\sum_{i=1}^{N_{T}} ( 
\hat y_{Ti}^{(h,s)}- \mu_{Ti})^2}{N_{T}},
    $$ 
    where $\hat{y}_{0i}^{(h,s)}={_h\hat{\bm{\beta}}_s\!\!}^\top \bm{x}_{0i}$, 
    $\hat{y}_{Ti}^{(h,s)}={_h\hat{\bm{\beta}}_s\!\!}^\top \bm{x}_{Ti}$, $\mu_{0i}=\bm{\beta}^\top \bm{x}_{0i}$,
$\mu_{Ti}=\bm{\beta}^\top \bm{x}_{Ti}$ and $_h\hat{\bm{\beta}}_s$ is the OLS estimate of ${\bm{\beta}}$ based on the subsample $s_{n}^{(h,s)}$;  
\item 
    The  standard error in the prediction set ${\bm{D}_0}$ and   in the test set  ${\bm{D}_T}$: 
    $$
    {\rm SE}_{\bm{D}_0}^{(h,s)} =\dfrac{\sum_{i=1}^{N_{0}} \big( \hat y_{0i}^{(h,s)} - y_{0i}\big)^2}{N_{0}} 
     \;\; {\rm and} \;\;
     {\rm SE}_{\bm{D}_T}^{(h,s)}=\dfrac{\sum_{i=1}^{N_{T}} \big( \hat y_{Ti}^{(h,s)} - y_{Ti}\big)^2}{N_{T}}
    .
    $$
\end{itemize}

\noindent
Table 1 displays the Monte Carlo averages, 
$$
{\rm MSPE}_{{\cal X}_0}=\dfrac{\sum_{h=1}^H \sum_{s=1}^S  {\rm MSPE}_{{\cal X}_0}^{(h,s)}}{HS} \;\; {\rm and} \;\;
{\rm Log(det)}=\dfrac{\sum_{h=1}^H \sum_{s=1}^S {\rm Log(det)}^{(h,s)}}{HS},
$$
for the different sampling strategies: non-inf. I, non-inf. D, inf. I, inf. D and SRS, respectively. 
The results are obtained  having setted $n=500$, $\tilde{N}=2 \cdot n$, $\nu_1=2$ and $\nu_2=3$.
\begin{table}[ht]
\centering
\begin{tabular}{rcc}
  \hline
 Algorithm &  ${\rm MSPE}_{{\cal X}_0}$ &  Log(det) \\ 
  \hline
non-inf. I & {\bf 0.0857}  & 93.4269 \\ 
non-inf. D & 0.0947  & {\bf 94.3877} \\
inf. I & 0.0938  & 92.0869 \\ 
inf. D & 0.1030  & 92.7748 \\ 
SRS & 0.2056  & 82.5234 \\ 
   \hline
\end{tabular}
\caption{{\footnotesize Monte Carlo averages ${\rm MSPE}_{{\cal X}_0}$ and Log(det) for the subsamples of size $n=500$ obtained from the different Algorithms
}}
\end{table}

Accordingly to the definitions of I- and D-optimality, the minimum value of the ${\rm MSPE}_{{\cal X}_0}$ is associated to the noninformative I-Algorithm, while the maximum value of the Log(Det) corresponds to the noninformative D-subsample. Therefore, Algorithms 1 and 4 provide samples that do not include high leverage points and are \lq\lq nearly" D- and I-optimal (they are not exactly D- and I-optimal because of the exclusion of these high values). 

Table 2 lists the following Monte Carlo averages: \\[.2cm]
${\rm SPE}_{{\cal X}_{0}}=\sum_{h=1}^H \sum_{s=1}^S {\rm SPE}_{{\cal X}_{0}}^{(h,s)}/HS\;$, 
$\;\;{\rm SPE}_{{\cal X}_{T}}=\sum_{h=1}^H \sum_{s=1}^S {\rm SPE}_{{\cal X}_{T}}^{(h,s)}/HS$, \\[.2cm]
${\rm SE}_{\boldmath{D}_0}=\sum_{h=1}^H \sum_{s=1}^S {\rm SE}_{\boldmath{D}_0}^{(h,s)}/HS\;$ 
and 
$\;{\rm SE}_{\boldmath{D}_T}= \sum_{h=1}^H \sum_{s=1}^S {\rm SE}_{\boldmath{D}_T}^{(h,s)}/HS$,\\[.2cm]
for the different subsamples. These Monte Carlo averages represent an empirical version of MSPE and MSE on ${\cal X}_{0}$ and ${\cal X}_{T}$, respectively. 
From these results, we can appreciate the prominent role of the informative procedures in selecting subsamples without outliers. In fact, when the database includes outliers in $Y$ which are not associated with high leverage points (as in this simulation study), then  only the informative procedure enables us to exclude these abnormal values from the subsample. 

\begin{table}[ht]
\centering
\begin{tabular}{rccccc}
  \hline
 Algorithm & ${\rm SPE}_{{\cal X}_{0}}$ & ${\rm SPE}_{{\cal X}_{T}}$ & ${\rm SE}_{\boldmath{D}_0}$ & ${\rm SE}_{\boldmath{D}_T}$ \\ 
  \hline
 non-inf. I & 6.5104 & 6.8020 & 16.0792 & 16.3538 \\  
 non-inf. D & 6.1011 & 6.2945 & 15.5982 & 15.7969 \\  
 inf. I & {\bf 0.1464} & {\bf 0.1494} & {\bf 9.4445} & {\bf 9.5337} \\ 
 inf. D & 0.1594 & 0.1601 & 9.4564 & 9.5448 \\ 
 SRS & 0.2629 & 0.2671 & 9.5683 & 9.6594 \\ 
   \hline
\end{tabular}
\caption{{\footnotesize Monte Carlo averages ${\rm SPE}_{{\cal X}_{0}}$, ${\rm SPE}_{{\cal X}_{T}}$, ${\rm SE}_{\boldmath{D}_0}$ and ${\rm SE}_{\boldmath{D}_T}$ for the subsamples of size $n=500$ obtained from the different Algorithms
}}
\end{table}

\noindent
{\bf Remark.}
Actually, to take into consideration the randomness of the SRS technique, we have drawn $N_{SRS}=50$ different independent SRSs from each dataset  $_h\bm{D}_{s}$, for $h=1, \ldots,H$ and $s=1,\ldots,S$; the Monte Carlo averages for SRS are based also on these additional observations.	
	


\begin{thebibliography}{99}
	
	\bibitem{atkinson2007} Atkinson, A., Donev, A., \& Tobias, R. (2007). Optimum experimental designs, with SAS (Vol. 34). Oxford University Press.
	
	\bibitem{Chatterjee1986influential} Chatterjee, S., \& Hadi, A.S. (1986). Influential Observations, High Leverage Points, and Outliers in Linear Regression. Statistical Sciences, 1(3), 379-416.
		
	\bibitem{deldossi2020} Deldossi, L., \& Tommasi, C. (2021). Optimal design subsampling from Big Datasets, Journal of Quality Technology, online first, https://doi.org/10.1080/00224065.2021.1889418.
	
	\bibitem{drovandi2017} Drovandi, C. C., Holmes, C., McGree, J. M., Mengersen, K., Richardson, S., \& Ryan, E. G. (2017). Principles of experimental design for big data analysis.
	Statistical science: a review journal of the Institute of Mathematical Statistics, 32(3), 385.	
	
	\bibitem{harman2019optimaldesign} Harman, R., \& Filova, L. (2019). OptimalDesign: A Toolbox for Computing Efficient Designs of Experiments.	
	
	\bibitem{hoaglin1978} Hoaglin, D. C., \& Welsch, R. E. (1978). The hat matrix in regression and ANOVA. The American Statistician, 32(1), 17-22. 
		
	\bibitem{koller2011sharpening} Koller, M., \& Stahel, W. A. (2011). Sharpening wald-type inference in robust regression for small samples. Computational Statistics \& Data Analysis, 55(8), 2504-2515.
		
	\bibitem{rencher2008linear} Rencher, A. C., \& Schaalje, G. B. (2008). Linear models in statistics. John Wiley \& Sons.
		
	\bibitem{wang2019} Wang, H., Yang, M., \& Stufken, J. (2019). Information-based optimal subdata selection for big data linear regression. Journal of the American Statistical Association, 114(525), 393-405.
		
	\bibitem{wang2018} Wang, H., Zhu, R., \& Ma, P. (2018). Optimal subsampling for large sample logistic regression. Journal of the American Statistical Association, 113(522), 829-844.
		
	\bibitem{Yu&Wang2022} Yu, J., \& Wang, H. Y. (2022). Subdata selection algorithm for linear model discrimination. Statistical Papers, to appear.
		
	\end{thebibliography}
\end{document}